\documentclass[journal,twocolumn,10pt,twoside]{IEEEtranTCOM}

\normalsize

\usepackage{cite}
\usepackage{setspace}

\usepackage{graphicx}
\graphicspath{%
    {Figures/}
    {converted_graphics/}
}

\usepackage[cmex10]{amsmath}
\usepackage{mathrsfs}
\usepackage{amsfonts}
\usepackage{amssymb}
\usepackage{amsthm}
\usepackage[caption=false,font=footnotesize]{subfig}
\usepackage{url}
\newtheorem{mylemma}{Lemma}
\newtheorem{mytheorem}{Theorem}

\begin{document}
%
\title{Orthogonal Time Frequency Space Modulation}

\author{R. Hadani, S. Rakib, S. Kons, M. Tsatsanis,  \\
A. Monk, C. Ibars, J. Delfeld, Y. Hebron, \\
A. J. Goldsmith, A.F. Molisch, and R. Calderbank
 \thanks{R. Hadani, S. Rakib, S. Kons, A. Monk, C. Ibars, J. Delfeld, Y. Hebron are, and M. Tsatsanis was, with Cohere Technologies, Inc., Santa Clara, CA 95051 USA. R. Hadani is also with the University of Texas, Austin, TX 78712. A. Goldsmith is with Stanford University, Stanford, CA 94305. A. F. Molisch is with the University of Southern California, Los Angeles, CA 90089. R. Calderbank is with Duke University, Durham, NC 27708. Parts of this paper were presented at IEEE WCNC 2017 \cite{hadani2017orthogonal}. }
 \\
\thanks{ This work has been submitted to the IEEE for possible publication.  Copyright may be transferred without notice, after which this version may no longer be accessible.}
 }

\maketitle


\begin{abstract}
This paper introduces a new two-dimensional modulation technique called Orthogonal Time Frequency Space (OTFS) modulation. OTFS has the novel and important feature of being designed in the delay-Doppler domain. When coupled with a suitable equalizer, OTFS modulation is able to exploit the full channel diversity over both time and frequency. Moreover, it converts the fading, time-varying wireless channel experienced by modulated signals such as OFDM into a time-independent channel with a complex channel gain that is essentially constant for all symbols. 

This design obviates the need for transmitter adaptation, and greatly simplifies system operation. The paper describes the basic operating principles of OTFS as well as a possible implementation as an overlay to current or anticipated standardized systems. 
OTFS is shown to provide significant performance improvement in systems with high Doppler, short packets, and/or large antenna array. In particular, simulation results indicate at least several dB of block error rate performance improvement for OTFS  over OFDM in all of these settings. 
\end{abstract}



\section{Introduction}\label{sec:intro}

4G cellular communications has achieved enormous success, due in particular to its ability to provide high data rates to a large number of users. 
Those data rate requirements were mainly driven by wireless video demand (around $70 $ \% of all cellular traffic \cite{index2015forecast}), usually consumed when the user is stationary ($>70 $ \% of all connections are with indoor users).  
Over the past years, interest has turned to the development of fifth-generation cellular communications \cite{andrews2014will}. It is anticipated that carrier investment will require new applications (beyond high-speed video connections), including Internet of things (IoT), and high-velocity V2X (vehicle-to-vehicle V2V and vehicle-to-infrastructure V2I) connections. 

Given the emergence of new applications, it is natural to ask whether 5G will benefit from a change in the modulation and multiple access method similar to earlier generation jumps, progressing from analog, to digital TDMA, CDMA, and OFDM. It is well known that OFDM is capacity-achieving in frequency-selective channels. However, {\em this optimality holds only under a set of very specific assumptions}, including (i) knowledge of the channel state information (CSI) at the transmitter, (ii) Gaussian modulation alphabet, (iii) long codewords (which implies the absence of latency constraints), and (iv) unlimited receiver complexity. These assumptions are not fulfilled in many of the 5G applications. It is thus imperative to investigate an {\em ab initio} design of modulation and multiple access for next-generation cellular applications. 


This paper introduces Orthogonal Time Frequency Space (OTFS) modulation, a new modulation scheme whereby each transmitted symbol experiences a near-constant channel gain even in channels with high Doppler, or at high carrier frequencies (mm-wave). Essentially, OTFS performs modulation in the delay-Doppler domain (also known as the {\em Zak domain} \cite{janssen1988zak}), which is naturally suited for transmission over time-varying wireless propagation channels. OTFS thus effectively transforms the time-varying multipath channel into a two-dimensional channel in the delay-Doppler domain. Through this transformation, coupled with equalization in this domain, all symbols over a transmission frame experience the same channel gain. Equivalently, OTFS can be interpreted as modulating each information symbol onto one of a set of two-dimensional (2D) orthogonal basis functions in the time-frequency plane that span the bandwidth and time duration of the transmission burst or packet. These spread basis functions allow, in conjunction with an appropriate equalizer, the extraction of the full channel diversity,\footnote{Full diversity could also be extracted in the time-frequency domain through an appropriately-designed equalizer, however the sparsity and lower variability of the channel in the delay-Doppler domain makes our approach less complex and more robust. } which leads to the almost-constant channel gain mentioned above; in a typical situation we might see a reduction of the standard deviation of power variations from 4 dB to 0.1 dB.  

The Zak representation of signals can be interpreted as a generalization of the time representation of signals on one hand, or the frequency representation of signals on the other hand. Thus, OTFS can be viewed as a generalization of OFDM or TDMA; for example it reduces to OFDM if the two-dimensional basis functions are subchannel carriers, or it can be seen as a generalization of single-carrier transmission, where the localization of the basis pulses is not only in the delay domain, but also the Doppler domain. Furthermore, since OTFS uses basis functions extending over the entire bandwidth and duration, OTFS is a generalization of (two-dimensional) CDMA; specifically a generalization of DFT-spread OFDM such that the spreading is not only one-dimensional along the frequency axis, but rather two-dimensional in the time-frequency plane. However, in contrast to CDMA and OFDM, the set of OTFS basis functions is specifically derived to combat the dynamics of the time-varying multipath channel. In summary, OTFS is designed to inherit advantageous properties from each of OFDM, TDMA, and CDMA. Typical gains are on the order of $2-4$ dB at $10$\% packet error rates (PERs), and higher for lower target PERs. 

The relatively constant channel gain over all symbol transmissions, which is one of the hallmarks of OTFS, greatly reduces the overhead and complexity associated with physical layer adaptation. It also presents the transport and application layer with a robust slowly varying channel, which is highly desirable when running over TCP/IP \cite{Diggavi_et_al_2004} and for the delay-sensitive applications envisioned for 5G. Moreover, full diversity enables linear scaling of throughput with the number of antennas, regardless of channel Doppler. In addition to OTFS's full diversity benefits, since the delay-Doppler channel representation is very compact, OTFS enables dense and flexible packing of reference signals, a key requirement to support the large antenna arrays used in massive MIMO.  

\subsection{Related papers}

The delay-Doppler representation of {\em signals} goes back to work in mathematics and physics \cite{zak1967finite}; an excellent tutorial introduction to the Zak transform can be found in \cite{janssen1988zak}. The delay-Doppler representation of time-varying {\em channels}  is described in-depth in the seminal work of Bello \cite{bello1963characterization}; the generalization to directional time-varying channels, relevant to multiple-antenna systems, was given in \cite{steinbauer2001double}, \cite{kattenbach2002statistical}, and \cite{sayeed2002deconstructing}. 

Since the 1990s, a variety of papers have suggested the use of time-frequency diversity transmission. Refs. \cite{sayeed1999joint,sayeed1998multiuser} established a signal model that presents the received signal as a canonical decomposition into delay and Doppler shifted versions of a basis signal, and suggests a delay-Doppler RAKE receiver that exploits the dispersion in both dimensions. Extensions of these ideas to the multi-antenna case appear in \cite{sayeed2002deconstructing}. Ref. \cite{ma2002maximum} points out that a time-frequency Rake receiver does not obtain optimal diversity as it is not optimized on the transmit side, and designs linear precoders that obtain full diversity order in doubly selective channels. Training strategies for block precoders are described in \cite{ma2003optimal}. Guard intervals in the frequency domain are designed in \cite{dean2017new}. The approaches in these works all differ from OTFS in that their system designs are in the time-frequency domain rather than the delay-Doppler domain. 

Other authors have investigated channel estimation in delay-Doppler channels and the impact of imperfect CSI: \cite{he2008doubly,tugnait2010multiuser} propose a basis expansion model using discrete prolate spheroidal sequences; 
\cite{zhou2015energy} showed that imperfect CSI need not lead to a reduction of the Doppler diversity order. Note, however, that in this paper we will mostly assume perfect knowledge of the CSI at the receiver. 

A separate body of prior work focuses on time-frequency pulse shape design  to minimize dispersion after transmission through the channel, based on Gabor system theory. A special case of this body of work is pulse-shaped OFDM.  Various criteria for pulse shape optimization have been considered in earlier works including suppressing ISI and ICI \cite{kozek1998nonorthogonal}, maximizing signal-to-interference-plus-noise ratio (SINR) \cite{das2007max}, or optimizing spectral efficiency through the use of non-rectangular lattices in the time-frequency domain \cite{strohmer2003optimal}.  These works all differ from OTFS in that they attempt to mitigate or fully remove the ISI and ICI through pulse shape design in the time-frequency domain. OTFS, in contrast, is designed so that its information symbols experience minimal cross-interference as well as full diversity through appropriate design of the modulation lattice and pulse shape in the delay-Doppler domain. While Gabor system theory tries to create ambiguity functions\footnote{The ambiguity function of a signal $g(t)$, denoted by $A_g(\tau,\nu)$, is a 2D generalization of the classical 1D auto-correlation function, measuring the correlation of the signal with a copy of itself, time and frequency shifted by $\tau$ and $\nu$ respectively. The classical auto-correlation function of $g(t)$ coincides with $A_g(\tau,0)$} with a single peak at zero and vanishing along a lattice -  a shape consistent with the Balian-Lowe theorem - we will see later on that OTFS in contrast creates ambiguity functions that have a much sharper peak at zero, but repeat periodically along the lattice. 

OTFS was first described in our conference paper \cite{hadani2017orthogonal}, and this has inspired a number of follow-up works by different groups. Refs. \cite{raviteja2017low,li2017simple,zemen2017low,murali2018otfs} propose various types of simplified receiver structures, usually based on iterative approaches. Some discussion of a discrete signal model, modulator design, and performance analysis is given in \cite{farhang2017low,rezazadehreyhani2017analysis}. Discussion of the diversity order achievable with different block coder designs (including OTFS) is provided in \cite{zemen2017orthogonal}. The current paper provides a more comprehensive theoretical development than in \cite{hadani2017orthogonal} and we hope that it will encourage further exploration of various system attributes by the community.


\subsection{Remainder of the paper}

In Section \ref{sec:channel} we describe the wireless channel in terms of its delay-Doppler characteristics, for which OTFS is designed; the remainder of Sec. II provides mathematical preliminaries including a general framework for time-frequency modulation. 
Section \ref{sec:modulation} develops the details of OTFS as a modulation that matches wireless channel characteristics through two processing steps.
Next, we discuss the interpretation of OTFS, and its applications for various 5G scenarios. 
Section \ref{sec:results} presents performance results for OTFS coupled with equalization, demonstrating its advantages over OFDM in high Doppler channels, with short packets, 
and with MIMO arrays. 
The paper concludes in Section \ref{sec:conclusion}.


\section{Delay-Doppler representation of channels and signals}

This section presents the delay-Doppler representation of channels and signals, as well as a general mathematical description of modulation/demodulation in the time-frequency domain. The results established here will be used in Sec. III to give a concise explanation of OTFS and relate it to other modulation formats.  

\subsection{The Delay-Doppler Channel}\label{sec:channel}
It is well known since the classical paper of Bello \cite{bello1963characterization} that a time-varying propagation channel can be represented by either its time-varying impulse response, the time-varying transfer function, or the Doppler-variant impulse response.\footnote{A fourth representation, the Doppler-variant 
transfer function, is rarely used} Among these representations, the Doppler-variant impulse response is a natural fit to the propagation physics. Specifically, the complex baseband Doppler-variant impulse response $h_{\rm c}(\tau,\nu)$ characterizes the channel response to an impulse, at delay $\tau$ and Doppler $\nu$ \cite{jakes1974mobile}. 
The received signal due to an input signal $s(t)$ transmitted over this channel
is given by:
\begin{equation}\label{eq:doppler}
  r(t) = \iint {h_{\rm c}(\tau,\nu)e^{j2\pi\nu( t - \tau)}s(t - \tau)\,\mathrm{d}\tau \,\mathrm{d}\nu}.
\end{equation}

According to (\ref{eq:doppler}) the received signal is a superposition of reflected copies of the transmitted signal, where each copy is delayed by the path delay $\tau$, frequency shifted by the Doppler shift $\nu$ and weighted by the time-independent complex valued delay-Doppler impulse response $h_{\rm c}(\tau,\nu)$ for that particular $\tau$ and $\nu$. Due to the connection to the physical scatterers, the representation is physically meaningful and, under the assumption of single-scattering processes, the location of scatterers can be found directly from it. Typical Doppler shifts are on the order of 10 Hz - 1 kHz, though larger values may occur in scenarios with extremely high mobility (e.g., high-speed trains) and/or high carrier frequency.  

{\em Remark 1}: there exist two different interpretations of the Doppler-variant impulse response, which differ by a term $e^{j 2 \pi \nu \tau}$. The difference can be interpreted as the question of whether we first apply the delay shift and then the Doppler shift, or vice versa. In any case, as long as the notation is consistent, equivalent results can be obtained with either definition. Note that this ordering issue is equivalent to definitions of the time-varying impulse response that can be either the response of a system to a delta pulse at time $t$ or at time $t-\tau$ \cite{kattenbach1997charakterisierung}. 

A further important attribute of the delay-Doppler channel representation $h_{\rm c}(\tau,\nu)$  is its compactness and sparsity. Since typically there is only a small number of physical reflectors with associated Dopplers, far fewer parameters are needed for channel modeling and estimation in the delay-Doppler domain than in the time-frequency domain.\footnote{Of course, a representation with the same number of parameters can be found also in the time-frequency domain, but the representation is not naturally obvious, and as a matter of fact is often based on transformation to the delay-Doppler domain.} This sparse representation for typical channel models, including those in LTE, has important implications for channel estimation/prediction and tracking and also for taming complexity of channel equalization/pre-coding in high order MIMO and MU-MIMO systems.

\subsection{The Heisenberg transform and twisted convolution}

Conceptually,  Eq. \eqref{eq:doppler} can be interpreted as a linear operator $\Pi_h(\cdot)$ that is parameterized by the impulse response function $h=h_{\rm c}(\tau,\nu)$ and that is operating on the input signal  $s(t)$ to produce the output signal $r(t)$, that is:
\begin{equation}\label{eq:hilbert}
  \Pi_h(s):\hspace{1em}s(t) \xrightarrow{\Pi_h}r(t).
\end{equation}
In the mathematics literature, the operator parameterization $h \rightarrow \Pi_h$ is called the {\em Heisenberg transform} which can be viewed as a non-commutative generalization of the Fourier transform  \cite{mecklenbrauker1989tutorial}.

As we will see below, multi-carrier modulations also utilize the Heisenberg transform on the transmitted symbols, hence the received signal is a composition (cascade) of two Heisenberg transforms, one corresponding to the modulation, and the other corresponding to the channel. With this in mind, the main technical statement about the Heisenberg transform is the twisted convolution property which can be viewed as a non-commutative generalization of the convolution property of the Fourier transform. The precise statement is summarized in the following Lemma.

\begin{mylemma} {\bf Twisted convolution property.} \label{Twisted_Lemma} Let $\Pi_{h_1}$ and $\Pi_{h_2}$ be two Heisenberg operators parameterized by functions $h_1(\tau,\nu)$ and $h_2(\tau,\nu)$, as defined in (\ref{eq:doppler},\ref{eq:hilbert}), applied in composition to a signal $s(t)$. Then we have: 
\begin{equation}\label{eq:Heisenbergcascade}
\Pi_{h_2} \left (\Pi_{h_1} (s(t)) \right ) =  \Pi_h (s(t)) .
\end{equation}
where $h(\tau,\nu)=h_2(\tau,\nu) \ast_{\sigma} h_1(\tau,\nu)$ is the {\em twisted convolution} of $h_1(\tau,\nu)$ and $h_2(\tau,\nu)$, 
defined as:
\begin{equation}\label{eq:conv_mod}
   h(\tau,\nu)=\iint{h_2(\tau^\prime,\nu^\prime) h_1(\tau-\tau^\prime,\nu-\nu^\prime)e^{j2\pi \nu^\prime (\tau-\tau^\prime)}}\mathrm{d}\tau^\prime \mathrm{d}\nu^\prime.
\end{equation}
\end{mylemma}

\begin{proof}
Let
\begin{equation}\label{eq:g1}
  r_1(t) = \iint h_1(\tau,\nu)e^{j2\pi \nu(t-\tau)}s(t-\tau)\mathrm{d}\tau \mathrm{d}\nu
\end{equation}
\begin{equation}\label{eq:g2}
  r(t) = \iint h_2(\tau,\nu)e^{j2\pi \nu(t-\tau)}r_1(t-\tau)\mathrm{d}\tau \mathrm{d}\nu
\end{equation}
Substituting \eqref{eq:g1} into \eqref{eq:g2} we obtain after some algebra manipulations:
\begin{equation}
  r(t) = \iint h(\tau,\nu)e^{j2\pi \nu(t-\tau)} s(t-\tau)\mathrm{d}\tau \mathrm{d}\nu
\end{equation}
with $h(\tau,\nu)$ given by \eqref{eq:conv_mod}.
\end{proof}

\subsection{Time-Frequency Modulation}\label{sec:tfmodulation}
All time-frequency modulations (also termed multi-carrier modulations) can be cast in a unified framework consisting of the following components:
\begin{itemize}
  \item A lattice (also termed a grid) $\Lambda$ in the time-frequency domain that samples the time and frequency axes at integer multiples of $T$ and $\Delta f$ respectively, that is: 
  \begin{equation}\label{eq:lattice}
    \Lambda = \lbrace (nT, m \Delta f): n,m \in \mathbb{Z}\}.
  \end{equation}
  \item A packet burst with total duration of $NT$ seconds and total bandwidth of $M\Delta f$ Hz.
  \item A 2D sequence of modulated symbols $X[n,m]$ that we wish to transmit over a given packet burst, parametrized along a finite number of points of the lattice $\Lambda$ with indices $n=0,\dotsc,N-1$ and $m=0,\dotsc,M-1$.
  \item A transmit pulse $g_{\rm tx}(t)$ and associated receive pulse $g_{\rm rx}(t)$
whose inner product is bi-orthogonal with respect to translations by integer multiples of time $T$ and frequency $\Delta f$, that is:
  \begin{equation}\label{eq:orthog_cond}
\int{ e^{-j2\pi m\Delta f(t-nT)} g_{\rm rx}^*(t-nT)g_{\rm tx}(t)\mathrm{d}t}
   = \delta(m) \delta(n).
  \end{equation}
\end{itemize}
Note that the bi-orthogonality property \eqref{eq:orthog_cond} of the pulse shapes ensures that cross-symbol  interference (aka interference among adjacent grid points) is eliminated in symbol reception, as will be shown in the next subsection. 

{\bf Time-frequency modulator:} A time-frequency modulator with these components maps the 2D symbol sequence $X[n,m]$ defined on the lattice $\Lambda$ to a transmitted signal $s(t)$ defined as a superposition of delay-and-modulate operations applied to the transmit pulse $g_{\rm tx}(t)$, as follows:
\begin{equation}\label{eq:superposition}
  s(t)=\sum_{m=0}^{M-1} \sum_{n=0}^{N-1}X[n,m]e^{j2\pi m\Delta f(t-nT)}g_{\rm tx}(t-nT).
\end{equation}
Note that the modulation rule (\ref{eq:superposition}) coincides with the Heisenberg transform of the symbol sequence $X[n,m]$ applied to the transmit pulse. This can be viewed as a two-dimensional generalization of the OFDM modulation transform that maps modulated symbols multiplexed in the frequency domain (i.e. on each subcarrier) to the transmitted signal defined in the time domain. In the same manner that the channel operation \eqref{eq:doppler} can be interpreted as a Heisenberg operator \eqref{eq:hilbert} parameterized by the Doppler-varying impulse response applied to the transmitted signal, the  modulation rule (\ref{eq:superposition}) can also be interpreted as a Heisenberg operator $\Pi_X(\cdot)$ parameterized by the symbol sequence $X[n,m]$ that is applied to the transmit pulse shape $g_{\rm tx}(t)$, that is:  
\begin{equation}\label{eq:heisenberg}
  s(t) =\Pi_X(g_{\rm tx}).
\end{equation}
This interpretation is useful because it allows us to consider the received signal as a composition of two 
Heisenberg operators, one associated with the modulation rule and the other associated with the channel. The composition of the Heisenberg operators corresponding to the channel, \eqref{eq:doppler}, and to the modulation, \eqref{eq:heisenberg}, yields the received signal in the form $r(t)=\Pi_{h_c} (\Pi_X(g_{\rm tx}))+\Tilde{v}(t)$ where $\Tilde{v}(t)$ is additive noise at the receiver input. By applying the twisted convolution property of Lemma \ref{Twisted_Lemma}, we have that $\Pi_{h_c} (\Pi_X(g_{\rm tx}))=\Pi_{h_c \ast_{\sigma} X}(g_{\rm tx})$, hence $r(t)$ can be written more explicitly as:
\begin{equation}\label{eq:heis_def}
  \begin{split}
    r(t)=\iint{f(\tau,\nu)e^{j2\pi \nu (t-\tau)}g_{\rm tx}(t-\tau)}\mathrm{d}\tau \mathrm{d}\nu + \Tilde{v}(t),
  \end{split}
\end{equation}
where $f(\tau,\nu)$ is the twisted convolution of the continuous function $h_{\rm c}(\tau,\nu)$ with the discrete function $X[n,m]$, i.e., $f(\tau,\nu)=h_{\rm c}(\tau,\nu) \ast_{\sigma} X[n,m]$, which can be written more explicitly as: 

\begin{eqnarray}\label{eq:X_conv_h}
  \lefteqn{f(\tau,\nu)=} \nonumber \\ 
& & \hspace{-.3in} \sum_{n=0}^{N-1}\sum_{m=0}^{M-1} h_{\rm c}(\tau-nT,\nu-m\Delta f)X[n,m]e^{j2\pi(\nu-m\Delta f)nT}.\nonumber \\
&&
\end{eqnarray}
With this result established, we are now ready to examine the receiver processing steps.


{\bf Time-frequency demodulator:} The sufficient statistic for symbol detection is obtained by matched filtering of the received signal $r(t)$ with the {\em channel-distorted}, information-carrying pulses (assuming that the additive channel noise is white and Gaussian).\footnote{The sufficient statistic is obtained if the receiver correlator pulse shape is matched to the transmitter pulse, i.e., $g_{\rm rx} (t) = g_{\rm tx} (t)$. Instead, we consider a more general formulation where $g_{\rm rx} (t) \ne g_{\rm tx} (t)$ to accommodate cases such as addition of a cyclic prefix at the transmitter in OFDM.} The matched filter first computes the cross-ambiguity function between the received signal $r(t)$ and the receive pulse shape $g_{\rm rx}$. This function is denoted $A_{g_{\rm rx},r}(\tau,\nu)$ and is given by:
\begin{equation}
    A_{g_{\rm rx},r}(\tau,\nu)
\triangleq \int{   e^{-j2\pi \nu (t-\tau)}g_{\rm rx}^*(t-\tau)} r(t) \mathrm{d}t ~~.
\end{equation}
The cross-ambiguity function can be interpreted as a two-dimensional (delay-Doppler) correlation function.\footnote{Similar to Remark 1, an alternative definition of the cross-ambiguity function exists; we use here the definition consistent with (1) 
.} It is worth noting the central role the cross-ambiguity function plays in radar theory \cite{Stein_1981}, indicating the implicit link between communication theory and radar.\footnote{This link will become clear below once we introduce the OTFS basis functions, which, being concentrated in the delay-Doppler domain, are well suited (and, through appropriate choice of $g_{\rm tx}$ can be made optimum) for radar purposes \cite{auslander1985radar}}. The matched filter output is obtained by sampling the cross-ambiguity function along the points of the lattice $\Lambda$, i.e., at integer multiples of time $T$ and frequency $\Delta f$, yielding the 2D sequence:
\begin{equation}\label{eq:wigner}
  \hat{Y}[n,m]=A_{g_{\rm rx},r}(\tau,\nu) \rvert_{\tau=nT, \nu=m\Delta f}.
\end{equation}
The sampled cross-ambiguity function \eqref{eq:wigner} constitutes a transform  mapping the 1D continuous function $r(t)$ to the 2D sequence $\hat{Y}[n,m]$. This transform inverts the discrete Heisenberg transform and is referred to in the mathematics literature as the {\rm discrete Wigner transform}. The discrete Wigner transform can be viewed as a generalization of the OFDM de-modulator that is mapping a received OFDM signal to modulated symbols on the frequency grid. We now proceed to calculate the relationship between the matched filter output $\hat{Y}[n,m]$ and the transmitter input $X[n,m]$.

{\bf Time-frequency input-output relation:} We have already established in \eqref{eq:heis_def} that the input to the matched filter $r(t)$ can be expressed as the sum of a noise term $\Tilde{v}(t)$ and a signal term $\Pi_f(g_{\rm tx}(t))$ obtained as the Heisenberg operator parameterized by the impulse response $f(\tau,\nu)$ applied to the pulse shape $g_{\rm tx}(t)$. Consequently, the output of the matched filter, before sampling, is a sum of two terms:
\begin{equation}\label{eq:mf_out}
    \hat{Y}(t,f)
=A_{g_{\rm rx},\Pi_f(g_{\rm tx})}(\tau,\nu)_{\vert \tau=t,\nu=f}+A_{g_{\rm rx},\Tilde{v}}(\tau,\nu)_{\vert \tau=t,\nu=f}.
\end{equation}
The second term on the right side is the contribution of noise, which we denote by $V(t,f)$, while the first term is the matched filter output in the absence of noise, which we denote by $Y(t,f)$. Direct calculation reveals that the noise-free component can be expressed as twisted convolution of three terms:\footnote{while the following equation seems to mix notation from the time/frequency and the delay/Doppler domains, it is similar to the often-used notation of convolving a signal $s(t)$ with a system response $h(\tau)$ to provide an output $r(t)$.}
\begin{equation}\label{eq:cross_amb}
  Y(t,f)=h_{\rm c}(\tau,\nu)\ast_{\sigma}  X[n,m] \ast_{\sigma}  A_{g_{\rm rx},g_{\rm tx}}(\tau,\nu).
\end{equation}
Using this expression for the first term in \eqref{eq:mf_out} yields:
\begin{equation} \label{eq:e2e_channel_desc}
    \hat{Y}(t,f)
=h_{\rm c}(\tau,\nu)\ast_{\sigma}  X[n,m] \ast_{\sigma}  A_{g_{\rm rx},g_{\rm tx}}(\tau,\nu)+V(t,f).
\end{equation}
The matched filter output estimate of the modulation symbols is obtained by evaluating the continuous function $\hat{Y}(t,f)$ along the points of the lattice $\Lambda$, that is:
\begin{equation}\label{eq:mf_output_est}
  \hat{Y}[n,m]=\hat{Y}(t,f) \vert_{t=nT,f=m\Delta f}.
\end{equation}
To calculate the end-to-end input output relation, let us first consider the simple case of an ideal channel $h_{\rm c}(\tau,\nu)=\delta(\tau) \delta(\nu)$. In this case, direct calculation of the right hand side of (\ref{eq:e2e_channel_desc}) yields:
\begin{flalign}\label{eq:Y_conv_rel}
   \hat{Y}[n,m]=   \sum_{n^\prime=0}^{N-1}  & \sum_{m^\prime=0}^{M-1}  
  X[n^\prime,m^\prime] \nonumber \\ 
   \times &  A_{g_{\rm rx},g_{\rm tx}}((n-n^\prime)T,(m-m^\prime)\Delta f)+V[n,m],
\end{flalign}
where $V[n,m]=V(t,f) \vert_{t=nT,f=m\Delta f}$ is the sampling of the noise term along the lattice $\Lambda$. Invoking the bi-orthogonality condition (\ref{eq:orthog_cond}), we get in this case that:
\begin{equation}\label{eq:Y_X_V}
  \hat{Y}[n,m]=X[n,m]+V[n,m].
\end{equation}
We conclude that for an ideal channel the matched filter output perfectly recovers the transmitted input symbols up to an uncorrelated noise term. This is the generalization of the well-known perfect reconstruction property for OFDM in the presence of non-dispersive channels. 

Let us now consider the matched filter output for more general channels characterized by a non-trivial impulse response $h_{\rm c}(\tau,\nu)$. First, note that the impulse response has finite support bounded by the maximum delay and Doppler spreads $(\tau_\text{max}, \nu_\text{max})$ of the reflectors/scatterers. For simplicity, let us assume that the bi-orthogonality condition (\ref{eq:orthog_cond}) holds in a robust manner in the sense that the cross-ambiguity function vanishes in a neighborhood of each non-zero lattice point $(nT,m\Delta f)$ at least as large as the support of the channel response, that is, $A_{g_{\rm rx},g_{\rm tx}}(\tau,\nu)=0$ for $\tau \in (nT-\tau_\text{max}, nT+\tau_\text{max})$, $\nu \in (m\Delta f-\nu_\text{max}, m\Delta f+\nu_\text{max})$.\footnote{We will discuss below the situation when this robustness condition is violated.} Under these assumptions, direct calculation of the right hand side of (\ref{eq:e2e_channel_desc}) yields the following generalization of (\ref{eq:Y_conv_rel}):
\begin{multline}\label{eq:theorem2_proof}
  \hat{Y}[n,m] = \sum_{n^\prime=0}^{N-1} \sum_{m^\prime=0}^{M-1} X[n^\prime, m^\prime]\\
  \times \iint e^{-j2\pi \nu \tau} h_{\rm c}(\tau, \nu) A_{g_{\rm rx},g_{\rm tx}}((n-n')T-\tau,(m-m')\Delta f-\nu) \\
  e^{-j2\pi(m'\Delta f \tau -nT \nu)} \mathrm{d}\tau \mathrm{d}\nu + V[n,m].
\end{multline}
One can easily verify that due to the bi-orthogonality robustness condition, only the zero term $n'=n,m'=m$ survives in the right hand side of \eqref{eq:theorem2_proof} and as a result we obtain: 
\begin{equation}
\hat{Y}[n,m] = H[n,m] X[n,m]+V[n,m],
\end{equation}
where the complex gain factor $H[n,m]$ is given by:
\begin{equation}\label{eq:Y_H_X}
  H[n,m] = \iint{e^{-j2\pi \nu \tau}h_{\rm c}(\tau,\nu)e^{-j2\pi (m\Delta f\tau -nT\nu)}}\mathrm{d}\tau\mathrm{d}\nu.
\end{equation}
Observe that there is no cross-symbol interference affecting the sequence $X[n,m]$ in either time $n$ or frequency $m$, implying that the received symbol coincides with the transmitted symbol except for the multiplicative scale factor $H[n,m]$ (and additive noise term). This is the characteristic channel relation of an OFDM transmission through a time-invariant frequency-selective channel. It is important to note that if the cross-ambiguity function is not robustly bi-orthogonal then there is some cross-symbol interference. The bi-orthogonality and, in its absence, residual cross-symbol interference depends on the particular structure of the transmit and receive pulses $g_{\rm tx}$ and $g_{\rm rx}$. The complex gain factor $H[n,m]$ has an expression as a weighted superposition of Fourier exponential functions, (\ref{eq:Y_H_X}), revealing an interesting relation between the discrete time varying transfer function $H[n,m]$ and the delay-Doppler impulse response $h_c(\tau,\nu)$. This relation can be formally expressed via a two-dimensional transform called the symplectic Fourier transform. We devote the next subsection to the description of the symplectic Fourier transform and then we use it as the underlying building block of implementing OTFS as a time-frequency overlay modulation.


\subsection{The Symplectic Fourier Transform}

The Symplectic Fourier Transform is a variant of the 2D Fourier transform which is naturally associated with the Fourier kernel $e^{-j2\pi(m\Delta f\tau -nT\nu)}$ used in (\ref{eq:Y_H_X}) for converting between the delay-Doppler and time-frequency channel representations.
Specifically, we will focus on a finite version of the transform called the finite symplectic Fourier transform, which is denoted by SFFT. The input of the SFFT is a 2D periodic sequence $x_p[k,l]$ with periods $(M,N)$ and the output of the SFFT is a 2D periodic sequence $X_p[n,m]=\text{SFFT}(x_p[k,l])$ with periods $(N,M)$ (note that the period orders are reversed). The output and input sequences should be viewed as defined, respectively, along the points of the time-frequency lattice $\Lambda$, (\ref{eq:lattice}), and the reciprocal delay-Doppler lattice $\Lambda^\perp$ that samples the delay axis at integer multiples of $\Delta \tau = \frac{1}{M \Delta f}$ and the Doppler axis at integer multiples of $\Delta \nu = \frac{1}{NT}$, i.e.,:

\begin{equation}
\label{eq:Dual_Lat}
\Lambda^\perp=\lbrace (k\Delta \tau, l \Delta \nu): k,l \in \mathbb{Z} \rbrace.
\end{equation}
Note that the delay interval $\Delta \tau$ is inversely proportional to the burst bandwidth $M \Delta f$ and the Doppler interval $\Delta \nu$ is inversely proportional to the burst duration $NT$. Hence, increasing the burst duration/bandwidth increases the sampling resolution in delay/Doppler respectively. This is consistent with the principles of radar asserting that range/velocity resolution is proportional to the bandwidth/duration of the probing waveform.

The output sequence is given by the following Fourier summation formula:
\begin{equation}
\label{eq:SFFT}
    X_p[n,m] 
= \sum_{k=0}^{M-1} \sum_{l=0}^{N-1}  \! x_p[k,l]e^{-j2\pi (\frac{mk}{M}-\frac{nl}{N})}.
\end{equation}
Note that the SFFT couples the frequency variable with the delay variable and the time variable with the Doppler variable with a minus sign. This type of coupling is referred to in the mathematics literature as symplectic coupling. 

The inverse transform $ x_p[k,l]=\text{SFFT}^{-1}(X_p[n,m])$ is given by a similar summation formula:
\begin{equation}
\label{eq:ISFFT}
    x_p[k,l] 
= \frac{1}{MN}\sum_{n=0}^{N-1} \sum_{m=0}^{M-1} \! X_p[n,m]e^{-j2\pi (\frac{ln}{N}-\frac{km}{M})}.
\end{equation}
The main property of the SFFT is that it interchanges between circular convolution and point-wise multiplication of periodic sequences, analogous to the similar convolution property of the conventional finite Fourier transform. This statement is summarized in the following theorem.  
\begin{mytheorem} {\bf Symplectic convolution property.}
Let  $x_1[k,l]$ and $x_2[k,l]$ be periodic 2D sequences with periods $(M,N)$. In addition, let $X_1[n,m]=\text{SFFT}(x_1[k,l])$ and $X_2[n,m]=\text{SFFT}(x_2[k,l])$ be the corresponding Fourier transforms. The following relation holds:
\begin{equation}
\label{eq:prop2}
\text{SFFT}(x_1[k,l] \circledast  x_2[k,l])  =  X_1[n,m] X_2[n,m],
\end{equation}
where $\circledast$ denotes 2D circular convolution. 
\end{mytheorem}
\begin{proof}
Based on the definition of the SFFT, it is straightforward to verify that translation in delay-Doppler converts into a linear phase in time-frequency:
\begin{equation}
\text{SFFT}(x_2[k-k',l-l']) = X_2[n,m]e^{-j2 \pi \left(\frac{mk'}{M}-\frac{nl'}{N}\right)}.
\end{equation}
Based on this result we can evaluate the SFFT of a circular convolution as:
\begin{equation}
  \begin{split}
    &\text{SFFT} \left\{ \sum_{k'=0}^{M-1} \sum_{l'=0}^{N-1}x_1[k',l']x_2[(k-k') \text{mod}M, \,(l-l') \text{mod}N] \right\}\\
    &= \sum_{k'=0}^{M-1} \sum_{l'=0}^{N-1} x_1[k',l']X_2[n,m] e^{-j2 \pi \left(\frac{mk'}{M}-\frac{nl'}{N}\right)}\\
    &= X_1[n,m]X_2[n,m],
  \end{split}
\end{equation}
yielding the desired result.
\end{proof}

\section{OTFS Modulation}\label{sec:modulation}

We are now ready to describe OTFS modulation cast in the framework of Sec. II as a time-frequency multi-carrier modulation equipped with additional pre-processing transforming from the delay-Doppler domain to the time-frequency domain. 

\subsection{Interpretations of OTFS modulation}

Before going into the mathematical description of OTFS, we first describe the intuition behind it. OTFS can be described in several equivalent ways:
\begin{itemize}
\item {\em Modulation in the delay-Doppler domain:} Just as OFDM can be interpreted as carrying information over a compact ''basis pulse" in the time-frequency domain, OTFS can be interpreted as a dual that carries information over a compact basis pulse in the delay-Doppler domain. These delay-Doppler pulses are modulated with QAM (quadrature amplitude modulation) such that the QAM symbols carry the information. While the interaction of the OFDM waveform with the channel leads to a multiplication of the basis pulses with the time-varying transfer function, the interaction of the OTFS waveform with the channel leads to the dual operation, given by a two-dimensional convolution, of its basis pulse with the Doppler-variant impulse response (the two-dimensional Fourier dual of the time-varying transfer function). In the following we will mainly use this interpretation for the mathematical derivations.
\item {\em Spreading in the time-frequency domain:} OTFS can also be viewed purely in the time-frequency domain as a spreading scheme that carries information over non-compact (as a matter of fact, maximally spread-out) orthogonal basis functions in the time-frequency domain, see Fig. \ref{fig:basisfunctions}. With that interpretation, OTFS becomes a two-dimensional version of CDMA. Considerations about achievable diversity can be most easily obtained from this interpretation.\footnote{The lattice description of OTFS also provides a natural connection to CDMA.} 
\item{\em Zak representation:} A more canonical description of OTFS casts it as the modulation format naturally associated with a fundamental transform referred to in the mathematics literature as the Zak transform, analogous to the fact that OFDM is the modulation format naturally associated with the Fourier transform. An elaborate derivation along these lines is given in\cite{OTFS_Zak}. 
\end{itemize}

\begin{figure}
  \centering
  \includegraphics[scale=.3]{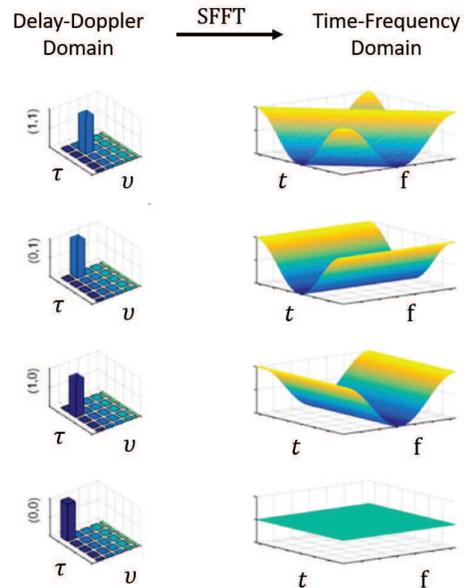} 
  \caption{2D basis functions in the Information (delay-Doppler) domain (left), and the corresponding symplectic Fourier dual basis functions in the time-frequency domain (right).}
  \label{fig:basisfunctions}
\end{figure}

\subsection{OTFS Modulation and Demodulation} \label{sec:OTFS-2D-transform}

OTFS modulation is a composition of two transforms at both the transmitter and the receiver, as shown in Fig. \ref{fig:otfs_ofdm}. The transmitter first maps the 2D sequence of information symbols $x[k,l]$ residing along the points of the reciprocal lattice $\Lambda^\perp$, see (\ref{eq:Dual_Lat}), in the delay-Doppler domain to a 2D sequence of complex numbers $X[n,m]$ residing along the points of the lattice $\Lambda$, (\ref{eq:lattice}), in the time-frequency domain. This is done through a combination of the finite symplectic Fourier transform and windowing. We call this composition of operations the {\em OTFS transform}. Next the Heisenberg transform is applied to the sequence $X[n,m]$ to convert the time-frequency modulated symbols to the time domain signal $s(t)$ for transmission over the channel. The reverse operations are performed in the receiver, mapping the received time signal $r(t)$ first to the time-frequency domain through the discrete Wigner transform, and then, via the inverse finite symplectic Fourier transform, to the delay-Doppler domain for symbol detection.  

\begin{figure}
  \centering
  \includegraphics[scale=.2]{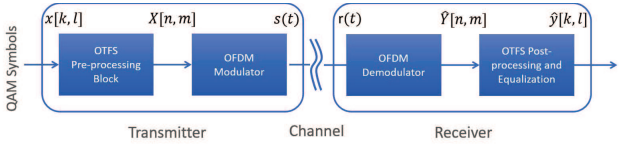} 
  \caption{OTFS Modulation Block Diagram: Transmitter and Receiver}
  \label{fig:otfs_ofdm}
\end{figure}

\textbf{OTFS delay-Doppler modulator}: Consider a finite 2D sequence of QAM information symbols $x[k,l]$, where $k=0,\dotsc, M-1$ and  $l=0,\dotsc, N-1$ that we wish to transmit. Let us denote by $x_p[k,l]$ the 2D periodized version of $x[k,l]$ with periods $(M,N)$. Further, let us assume a time-frequency modulation system defined by the lattice, packet burst, and bi-orthogonal transmit and receive pulses as described in Section~\ref{sec:tfmodulation}. In addition, let us assume a square summable transmit windowing function $W_{\rm tx}[n,m]$ that  multiplies the modulation symbols in the time-frequency domain. Given the above components, the OTFS modulated symbols are defined as follows:
\begin{equation}\label{eq:otfs_def}
    X[n,m] = W_{\rm tx}[n,m] \, \text{SFFT}(x_p[k,l]).
\end{equation}
The transform \eqref{eq:otfs_def} is comprised of the SFFT followed by a windowing operation in time-frequency and is referred to as the {\em OTFS transform}. The transmitted signal is obtained from the modulated sequence using the Heisenberg transform defined in \eqref{eq:heisenberg} as follows:
\begin{equation}
s(t) = \Pi_X(g_{\rm tx}(t)).
\end{equation}
 The composition of the OTFS transform and the Heisenberg transform comprises the OTFS modulation, as shown in the two transmitter blocks of Fig.~\ref{fig:otfs_ofdm}. An alternative interpretation expresses the output of the OTFS transform in the form:
\begin{equation}\label{eq:Xnm}
    X[n,m] =W_{\rm tx}[n,m] \sum_{k=0}^{M-1}\sum_{l=0}^{N-1} \! x[k,l] b_{k,l}[n,m],
 \end{equation}
 where $b_{k,l}[n,m]$ is the 2D sequence given by the sampled symplectic Fourier exponential function: 
 \begin{equation}   
    b_{k,l}[n,m] = e^{-j2\pi \left(\frac{mk}{M}-\frac{nl}{N}\right)}.
\end{equation}
This interpretation of \eqref{eq:Xnm} casts the OTFS transform as a CDMA like spreading system in time-frequency. It shows that each information symbol $x[k,l]$ is modulated by a 2D basis function $b_{k,l}[n,m]$ in the time-frequency domain, the shape of which is shown in Fig. \ref{fig:basisfunctions}. From this interpretation it is clear that every OTFS QAM symbol is spread over the full  time-frequency grid and hence is able to exploit the diversity associated with all the modes of the channel. 

\textbf{OTFS delay-Doppler demodulator}: For the definition of the demodulator let us assume a square summable receive windowing function $W_{\rm rx}[n,m]$ and consider a receive signal $r(t)$. The receive signal is demodulated as described in the following four steps:
\begin{enumerate}
  \item Apply the (discrete) Wigner transform to the signal $r(t)$ to obtain a time-frequency 2D sequence of demodulated symbols with unbounded support:
  \begin{equation}\label{eq:Y_A}
    \hat{Y}[n,m] = A_{g_{\rm rx},r}(\tau,\nu)\vert_{\tau=nT,\nu=m\Delta f}.
  \end{equation}
  \item Apply the receive window function $W_{\rm rx}[n,m]$ to the sequence $\hat{Y}[n,m]$ to obtain a shaped time-frequency 2D sequence of bounded support: 
    \begin{equation}
        \hat{Y}_W[n,m] = W_{\rm rx}[n,m]\, \hat{Y}[n,m]. 
    \end{equation}
  \item Periodize the sequence $\hat{Y}_W[n,m]$ to obtain a periodic time-frequency 2D sequence with periods $(N,M)$ along time and frequency respectively:  
    \begin{equation}
        \hat{Y}_p[n,m] = \sum_{n'=-\infty}^\infty \sum_{m'=-\infty}^{\infty} \hat{Y}_W[n-n'N,m-m'M].
    \end{equation}  
  \item Apply the inverse SFFT to the periodic time-frequency sequence $\hat{Y}_p[n,m]$ to obtain a periodic delay-Doppler sequence: 
    \begin{equation}\label{eq:xhat_sfft}
      \hat{y}_p[k,l] = \text{SFFT}^{-1}(\hat{Y}_p[n,m]).
    \end{equation}
\end{enumerate}
The output sequence of demodulated symbols is obtained as $\hat{y}[k,l]=\hat{y}_p[k,l]$ for $k=0,..,M-1$ and $l=0,..,N-1$. The last step can be interpreted as a projection of the time-frequency modulation symbols onto the two-dimensional orthogonal basis functions $b_{k,l}[n,m]$ as follows:
\begin{equation}\label{eq:xYb}
        \hat{y}[k,l] = \frac{1}{NM}\sum_{n=0}^{N-1} \sum_{m=0}^{M-1} \! \hat{Y}_p[n,m]b_{k,l}^*[n,m],
\end{equation}
where $b_{k,l}^*[n,m]$ stands for the conjugate 2D basis function: 
\begin{equation}
 b_{k,l}^*[n,m] = e^{-j2\pi \left(\frac{ln}{N}-\frac{km}{M} \right)}.
\end{equation}

{\bf OTFS delay-Doppler input-output relation:} We conclude this development with a description of the input-output relation between the periodic delay-Doppler modulated and demodulated sequences $x_p$ and $\hat{y}_p$ respectively. The relation is roughly given, up to an additive noise term, by convolution with the Doppler variant channel impulse response $h_c(\tau,\nu)$. In more precise terms, consider the periodic convolution of the channel impulse response with a filtering function as follows:\footnote{The window $w(\tau,\nu)$ is circularly convolved with the Doppler variant channel impulse response $h_{\rm c}(\tau,\nu)$ multiplied by the complex quadratic exponential $e^{-j2\pi \nu \tau}$.}
\begin{equation}\label{eq:window}
  h_w(\tau,\nu) = \iint{ e^{-j2\pi \nu' \tau'}h_{\rm c}(\tau',\nu') \, w(\nu-\nu', \tau-\tau')\mathrm{d}\tau' \mathrm{d}\nu'}.
\end{equation}
The filtering function $w(\tau,\nu)$ is a periodic function with periods $(M \Delta \tau,N \Delta \nu)$ in delay and Doppler respectively, obtained as the inverse discrete symplectic Fourier transform (denoted by SDFT) of a time-frequency window $W[n,m]$, that is:
\begin{equation}\label{eq:sfft_W}
  w(\tau,\nu) = \sum_{n=-\infty}^{\infty} \sum_{m=-\infty}^{\infty}e^{-j2\pi (\nu nT - \tau m \Delta f)}W[n,m],
\end{equation}
where $W[n,m]= W_{\rm tx}[n,m] \, W_{\rm rx}[n,m]$ is the product of the transmit and receive windows. Note that as the support of the window $W[n,m]$ along  time and frequency increases, the filtering function $w$ gets narrower in delay-Doppler and as a result $h_w(\tau,\nu)$ more closely approximates the true channel impulse response $h_c(\tau,\nu)$. The OTFS
input-output relation is summarized in the following key theorem.

\begin{mytheorem} {\bf OTFS delay-Doppler input-output relation.} \label{Main_Thm}
The input-output relation between the periodized demodulated noisy sequence $\hat{y}_p[k,l]$ and the periodized transmitted information symbol sequence $x_p[k,l]$ is given by:
\begin{equation}\label{xkl}
\begin{split}
  &\hat{y}_p[k,l] =\frac{1}{NM} \sum_{k'=0}^{M-1} \sum_{l'=0}^{N-1}h_w \left( k'\Delta \tau,l' \Delta \nu \right) \\
  & \times x_p[k-k',l-l']+v_p[k,l],
  \end{split}
\end{equation}
where $v_p[k,l]=\text{SFFT}^{-1}(V_p[n,m])$ and $V_p[n,m]$ stands for the periodization of the sampled and windowed time-frequency noise term $W_{\rm rx}[n,m]V[n,m]$.
\end{mytheorem}
\begin{proof}
See Appendix.
\end{proof}

A graphic representation of the Delay-Doppler input-output relation is depicted in Fig. \ref{fig:Channelconvolution}. 

\begin{figure}[htbp] 
  \centering
  \includegraphics[width=3.5in,keepaspectratio]{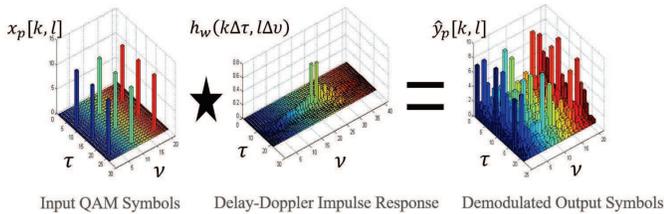}
  \caption{"Clean" OTFS delay-Doppler input-output relation (in the absence of noise): Convolution of (sampled and filtered) channel delay-Doppler impulse response with modulation symbols in the delay/Doppler domain.}
  \label{fig:Channelconvolution}
\end{figure}

\subsection{Equalization} \label{sec:equalizer}

For typical system parameters of broadband transmission, the assumption that the support of the channel impulse response is limited to the neighborhood of a lattice point, and thus 
 the bi-orthogonality condition, is not fulfilled. Consequently, the 2D intersymbol interference (\ref{xkl}) must be eliminated by a suitable equalizer. Possible structures include the 2D versions of standard equalizers, namely (i) linear equalizers, (ii) non-linear equalizers such as decision feedback and maximum-likelihood sequence estimators, possibly approximated as turbo equalizers. While a detailed discussion of equalizer structures for OTFS is beyond the scope of this paper, we note that linear equalizers generally perform poorly, see also Sec. V. More detailed discussions are given in  \cite{raviteja2017low,li2017simple,zemen2017low,murali2018otfs}. The impact of advanced equalizer structures on receiver complexity is limited due to the sparsity and translation invariance of the channel response, and the relative impact on {\em overall} receiver complexity (including decoding) is even less pronounced, especially since many modern systems use maximum-likelihood receivers.  


\section{Interpretation and implementation}

\subsection{Implementation as overlay}

OTFS can be implemented in a variety of ways. One expedient method is as an overlay of an existing OFDM system, since highly optimized hardware already exists for such systems, especially in the context of cellular (3GPP) and wireless LAN (WiFi) systems. Figure \ref{fig:overlay} shows the flow diagram that makes use of this structure. Current OFDM transceivers already implement a form of the Heisenberg/Wigner transform. At the transmitter, it is thus sufficient to perform a 2D OTFS transform (which can be implemented as a 2D SFFT), and let the resulting symbols be the input for the existing OFDM modulator. At the receiver, the output of the (soft) OFDM demodulator also undergoes a 2D OTFS transform, the results of which are used as input to the OTFS equalizer and demodulator. This can be thought of as a generalization of the approach in single-carrier (SC)-FDMA (also referred to as DFT-spread-OFDM) where a one-dimensional FFT is applied. 

\begin{figure}[htbp] 
  \centering
  \includegraphics[width=2.75in,keepaspectratio]{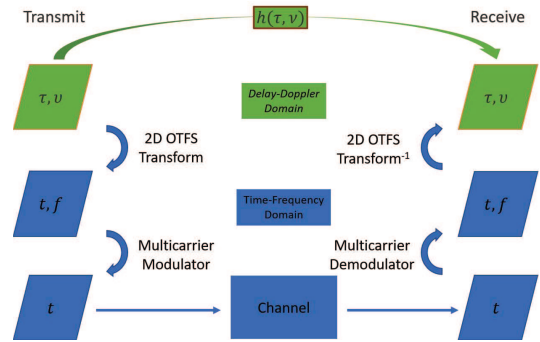}
  \caption{Signal flow in an implementation of OTFS as time-frequency overlay.}
  \label{fig:overlay}
\end{figure}

In terms of implementation complexity, we have to consider two components: the Heisenberg transform is already implemented in today's systems in the form of OFDM/OFDMA,  corresponding to a prototype filter $g(t)$ which is a square pulse. Other filtered OFDM and filter bank variations have been proposed for 5G \cite{banelli2014modulation}, which can also be accommodated in this general framework with different choices of $g(t)$. The second step of OTFS is the 2D finite symplectic Fourier transform (SFFT).  We now compare that complexity to that of SC-FDMA: for a frame of $N$ OFDM symbols consisting of $M$ subcarriers each, SC-FDMA adds $N$ DFTs of $M$ points each (assuming the worst case that all subcarriers are given to a single user). For these parameters, the additional complexity of SC-FDMA is then $MN\log_2(M)$ over the baseline OFDM architecture. For OTFS, the 2D SFFT has complexity $MN\log_2(MN)=MN\log_2(M)+MN\log_2(N)$, so the term $MN\log_2(N)$ is the OTFS additional complexity compared to SC-FDMA. For an LTE subframe with $M=1200$ subcarriers and $N=14$ symbols, the additional complexity is 37\% more compared to the additional complexity of SC-FDMA. However, note that this complexity increase only occurs if OTFS is implemented on top of an existing transceiver. Just like for SC-FDMA, some operations of the preprocessing and the IFFT of the OFDM modulator cancel each other out. In fact, one can show that an optimized OTFS transmitter has a complexity that is essentially half of that of an equivalent OFDM transmitter. 

\subsection{Multiplexing}

There are a variety of ways to multiplex several uplink or downlink transmissions in one OTFS frame. The most natural one is multiplexing in the delay-Doppler domain, such that different sets of OTFS basis functions, or sets of information symbols or resource blocks are given to different users. Given the orthogonality of the basis functions, the users can be separated at the receiver. For the downlink, the UE (user equipment) need only demodulate the portion of the OTFS frame that is assigned to it. This approach is the natural dual to resource block allocation in OFDMA. It is noteworthy, however, that the OTFS signals from all users extend over the whole time-frequency window, thus providing full diversity; by this we mean that for a channel with $Q$ clustered reflectors ($Q$ multipath components separable in either the delay or Doppler dimension) the OTFS modulation can achieve a diversity order equal to $Q$. Furthermore, this full spreading is also advantageous from a Peak to Average Power Ratio (PAPR) point of view (see Sec. \ref{sec:PAPR}). In the uplink direction, transmissions from different users experience different channel responses. Hence, the different subframes in the OTFS domain will experience a different channel. This can potentially introduce inter-user interference at the edges where two user subframes are adjacent, and would require guard gaps to eliminate it. 

An alternative is multiplexing in the time-frequency domain, i.e., different resource blocks or subframes are allocated to different users. These resource blocks can either be contiguous, or interleaved. In the former case, each user's signal is transmitted over a subset of the time-frequency plane, and thus has reduced diversity. 

\subsection{Diversity and channel gain}

From \eqref{xkl} we see that over a given frame, each demodulated symbol $\hat{x}[k,l]$ for a given $k$ and $l$ experiences
the same channel gain on the transmitted symbol $x[k,l]$. This combined with a non-linear equalizer at the receiver, as discussed in Sec. \ref{sec:equalizer}, allows to extract the full channel diversity. 
The almost-constant channel offers several important performance benefits. Firstly, it obviates the need for fast adaptive modulation and coding (AMC). While AMC provides significant benefits in slowly varying channels \cite{goldsmith1998adaptive}, it might either require high feedback overhead, or be completely impossible (when the channel coherence time is less than the feedback time) in systems operating in fast-varying channels. If obtaining occurate CSI (channel state information) knowledge at the TX becomes impossible, then channel variations are detrimental to performance, as the AMC must be chosen to accommodate the worst-case SNR. 
Thus, for high-mobility situations (vehicle-to-vehicle, high-speed train, etc.), channel whitening through spreading, and thus operating with a constant SNR over extended time periods, is not only the simpler solution, but also provides the better performance. 

The importance of a robust and fixed-rate channel is increased for applications that - due to latency constraints - do not allow retransmissions. This is especially critical for running applications over the TCP/IP protocol, which dramatically backs off the rate when packet failures occur, and subsequently takes a long time to converge again to a higher rate.

A high diversity could also be achieved in OFDM, using suitable interleaving and coding. However, that solution is subject to a number of drawbacks: (i) the information is not uniformly distributed over the time-frequency plane; thus diversity is not fully exploited, and the higher the coderate the more pronounced is the effect; (ii) it is not an effective solution for short codewords, since the diversity is upper bounded by the number of transmitted bits; in contrast OTFS always provides full diversity. 

A second important advantage of the almost-constant channel lies in enabling simplified equalizers and decoders, as well as precoders. For example, within the duration of an OTFS symbol, the equalizer coefficients do not have to be adapted, while every symbol in OFDM needs a different (though significantly simpler) equalizer.\footnote{ strictly speaking, the same equalizer coefficients can be applied to symbols that are within time and frequency intervals significantly smaller than the channel coherence time and coherence frequency, respectively. } More important, any signal predistortion can be done equally for all signals. 

The convergence to a constant channel gain can also be interpreted as "channel hardening", an effect well known from massive MIMO systems \cite{marzetta2016fundamentals}. Notably, this effect inherently occurs even in a single-antenna OTFS system. This can be explained by interpreting the antenna locations of the UE at different times during the considered time window as a massive ''virtual array".  

\subsection{PAPR} \label{sec:PAPR}

Low PAPR is an important goal for modulation/multiple access design since it reduces the maximum linear power requirements for the transmit amplifiers. This is particularly important for the uplink of cellular systems, since amplifiers in consumer devices such as handsets need to be low-cost. OTFS (considered here with delay/Doppler multiplexing) reduces uplink PAPR in two ways: (i) if a user is assigned a single Doppler frequency, then the PAPR is the same as for single-carrier transmission, i.e., significantly lower than for OFDM. (ii) in conjunction, due to the spreading operation, the packet transmission extends over a longer period of time than in OFDM which allows to increase the maximum energy per bit under Tx power constraints. This is particularly relevant for short packets. 

It is noteworthy that especially for short packets, OTFS achieves a superior trade-off between PAPR and performance compared to SC-FDMA, even in time-invariant channels.  While SC-FDMA can have low PAPR during the active signal duration, the {\em overall} PAPR is only small if the signal has a duty cycle close to unity, which in turn requires that (due to the small packet size) it utilizes only a single (or very few) subcarriers. However, such an approach, which is also used by LTE, leads to low frequency diversity and thus inferior performance; furthermore, for very short packets, SC-FDMA might still have a duty cycle $<1$ even for such a configuration. OTFS, on the other hand, can obtain full spreading in time and frequency while keeping the PAPR low.  


Standard SC-FDMA multiplexes data over contiguous bands. A more sophisticated variant is referred to as hopped SC-FDMA. This mode of transmission maximizes the link budget as it enjoys low PAPR comparable to single carrier and maximize transmission duration, while at the same time extracting additional diversity gain. However, there is a subtle phenomenon that renders this approach sub-optimal. To maintain low PAPR, the QAM order must be kept low - say QPSK. Under this constraint, the transmission rate can only be adjusted by changing the FEC rate, and thus the performance is governed by the restricted QPSK capacity (or restricted mutual information) instead of by the Gaussian capacity. Ref. \cite{carmon2015comparison} showed that in the presence of time-frequency selectivity the restricted capacity of multicarrier modulations is saturated, becoming strictly sub-optimal. 

\section{Performance Results}\label{sec:results}



In the following we present some results that demonstrate some key benefits of OTFS. The parameters of the simulations are chosen to be comparable to, and consistent with, a use of OTFS as an ``overlay" of a 4G LTE system (further advantages could be achieved by the use of an OTFS ``greenfield" system, e.g., by abolishing the cyclic prefix as discussed above). Thus, the performance advantages presented here can be seen as a lower bound on the performance gains relative to OFDM. 

To be more specific, the simulations use the  PHY layer parameters to comply with the 4G LTE specification (ETSI TS 36.211 and ETSI TS 36.212), unless otherwise specified. For the OTFS system we add the OTFS transform pre- and post-processing blocks at the transmitter and receiver respectively. We simulate the wireless fading channel according to the TDL-C channel model (delay spread of 300 ns, Rural Macro, and low correlation MIMO), one of the standardized channel models in 3GPP. The details of the simulation parameters used are summarized in Table I.



\begin{table}
\begin{center}
\tiny{
\begin{tabular}{|l|l|p{3.75in}}
  \hline
  {\bf Parameter} & {\bf Value} \\ \hline
  Carrier frequency (GHz) & 4.0 \\ \hline
  Duplex mode & FDD \\ \hline
  Subcarrier spacing (kHz) & 15 \\ \hline
  Cyclic prefix duration ($\mu$ s) & 4.7 \\ \hline
  FFT size & 1024 \\ \hline
  Transmission bandwidth (resource blocks) & 50 \\ \hline
  Antenna configuration & 1T 1R (SISO) \\ \hline
  Rank & Fixed rank \\ \hline
  MCS & fixed: 4QAM, 16 QAM, 64 QAM \\ \hline
  Control and pilot overhead & none \\ \hline
  Channel estimation & ideal \\ \hline
  Channel model & TDL-C, DS=300 ns \\ \hline
\end{tabular}
}
\caption{Simulation parameters for performance results}
\end{center}
\end{table}

Figure~\ref{fig:uncoded_BER} shows the BER of an uncoded system operating in a time-and-frequency dispersive channel, with a medium Doppler (velocity corresponding to highway vehicle speed) and with different modulation formats. When comparing the results of OFDM modulation to that of OTFS, when both are using MMSE equalizers, we find that for low modulation order (4QAM and 16 QAM), OTFS outperforms OFDM; in particular for 4QAM OFDM shows a significant error floor due to the intercarrier interference that is not present for OTFS. Furthermore, we see that the slope of the BER-vs-SNR curve is steeper for OTFS than for OFDM even outside the "error floor" region, which can be explained by the higher diversity order.
However, we also note that for higher-order modulation OFDM outperforms OTFS with MMSE (in line with the discussion in Sec. III.C). This can be remedied by the use of a non-linear equalizer in OTFS. We analyze in particular DFE structures, distinguishing between "standard" DFE and DFE with "genie"-aided feedback (i.e., for the interference subtraction the DFE knows the correct symbol sequence).   
We see that for 4QAM, there is at most a 1dB difference of the SNR when comparing genie-aided to standard DFE, indicating that error propagation is not a significant issue. However, error propagation can be more significant at higher modulation orders. 

\begin{figure}[htbp] 
  \centering
   \includegraphics[bb=0 0 560 420,width=2.7in,keepaspectratio]{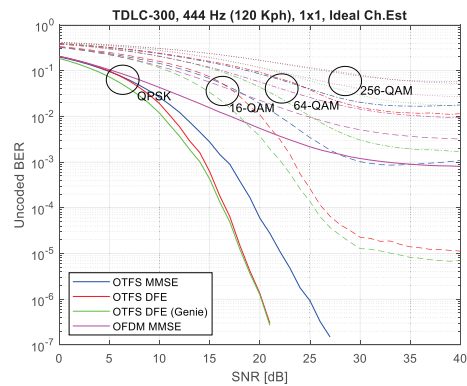}
  \caption{Uncoded BER for 4QAM/16QAM/64QAM/256QAM, TDL-C channel model, 120kmph. Curves show OFDM with MMSE equalizer, and OTFS with (i) MMSE equalizer, (ii) DFE equalizer with error propagation, and (iii) DFE equalizer with genie feedback}.
  \label{fig:uncoded_BER}

\end{figure}


%

Figure \ref{fig:BLER5} shows the coded error rate for a variety of equalizer structures. We again compare OFDM with MMSE equalization, to OTFS with MMSE and with DFE. We firstly see that MMSE equalization is not suitable for use in OTFS, and would lead to equal or worse performance than OFDM. We also see that due to error propagation, the standard DFE does not perform well either. These trends are much more pronounced at low code rate and high modulation order, consistent with our discussion of the uncoded case. However, an iterative DFE (which is seen to closely approximate the performance of a genie-aided DFE) performs very well, and OTFS with such a receiver significantly outperforms OFDM.

\begin{figure}[htbp] 
  \centering
  \includegraphics[width=3.5in,keepaspectratio]{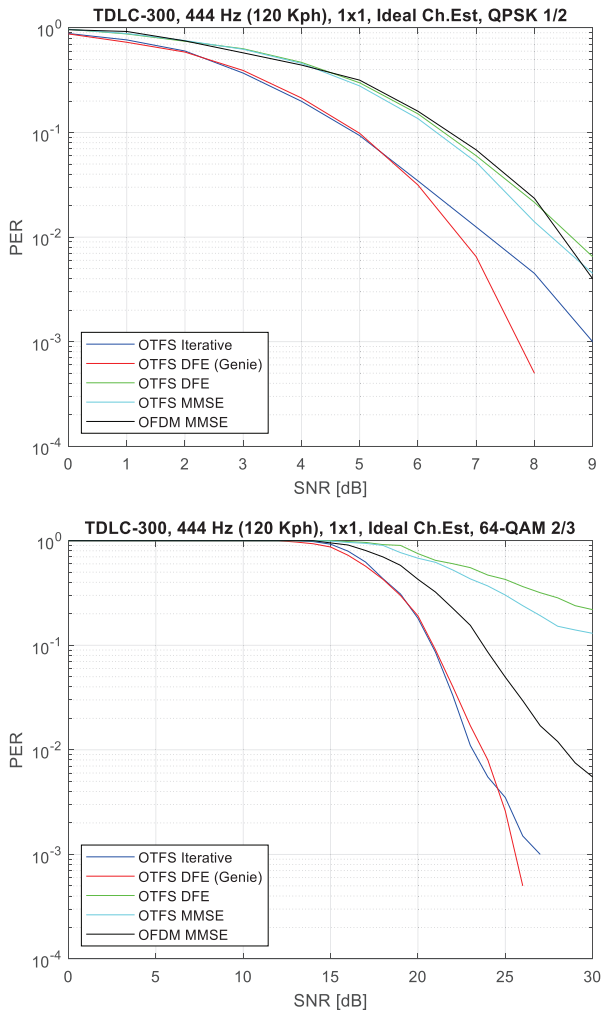}
  \caption{BLER for 4QAM, code rate 1/2 and 64QAM code rate 2/3, 120kmph.}
  \label{fig:BLER5}

\end{figure}

We next simulate a situation with short packets. Figure~\ref{fig:BLER2} shows the BLER performance for a system with mobile speeds of 30 km/h if only 4 resource blocks (48 subcarriers), called PRBs in LTE, are occupied by the user of interest out of a total of 50 resource blocks (600 subcarriers). This corresponds to a short packet length. Notice the increased diversity gain of OTFS compared to OFDM in this case resulting in gains of 4 dB or more as SNR increases. This diversity gain is because the OTFS transform spreads each QAM symbol over all time and frequency dimensions of the channel and then extracts the resulting full diversity, while OFDM limits the transmission to a narrow subchannel of 48 subcarriers.

\begin{figure}[htbp] 
  \centering
  \includegraphics[bb=0 0 560 420,width=2.7in,keepaspectratio]{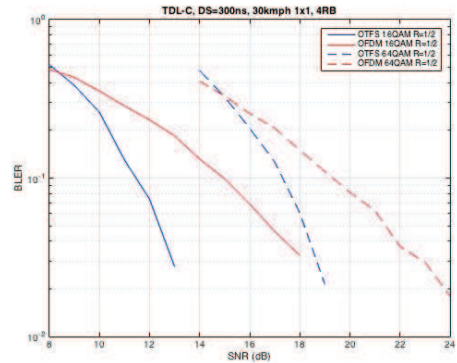}
  \caption{BLER  for short packet length (4 PRBs out of 50), 16QAM/64QAM, Code rate R=1/2, 30kmph.}
  \label{fig:BLER2}
\end{figure}

A similar situation arises when a number of users are multiplexed, leading to a smaller percentage of resources for each user. Figure \ref{fig:number_of_users} illustrates this effect. It shows the coded PER when the packet size per user is varied (corresponding to having to supply a different number of users). For OFDM, performance becomes worse as the packet size decreases, since it implies that the diversity of the system decreases. OTFS shows essentially unchanged behavior for packet sizes ranging from 2 to 50 PRB (physical resource blocks). 


\begin{figure}[htbp] 
  \centering
  \includegraphics[width=2.7in,keepaspectratio]{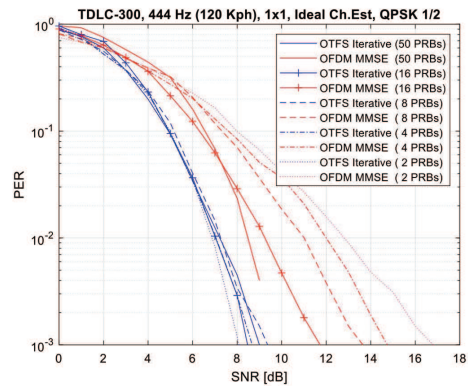}
  \caption{BLER  for QPSK with rate 1/2 code and different numbers of users.}
  \label{fig:number_of_users}
\end{figure}

As pointed out previously, one of the key advantages of OTFS is the stability of the effective SNR, which not only provides better operating points in systems without sufficiently fast CQI feedback, but also offers a more stable channel to the MAC layer. Figure~\ref{fig:SNR_distribution} shows an example for this effect. In an ETU channel with 120 km/h UE mobility, the SNR seen by an OFDM system shows significant variations (more than 4 dB standard deviation), while the SNR of an OTFS system with a 1 ms window has only 1.1 dB standard deviation, and OTFS with a 10 ms time window is almost constant (standard deviation 0.2 dB). When considering the cdf of the SNR, if an outage probability of $0.01$ is required, then an OFDM system requires a fading margin that is 7 dB larger than OTFS with a 1 ms window, and 9 dB compared to OTFS with a 10 ms window. 

\begin{figure}[htbp] 
  \centering
  \includegraphics[width=3.2in,keepaspectratio]{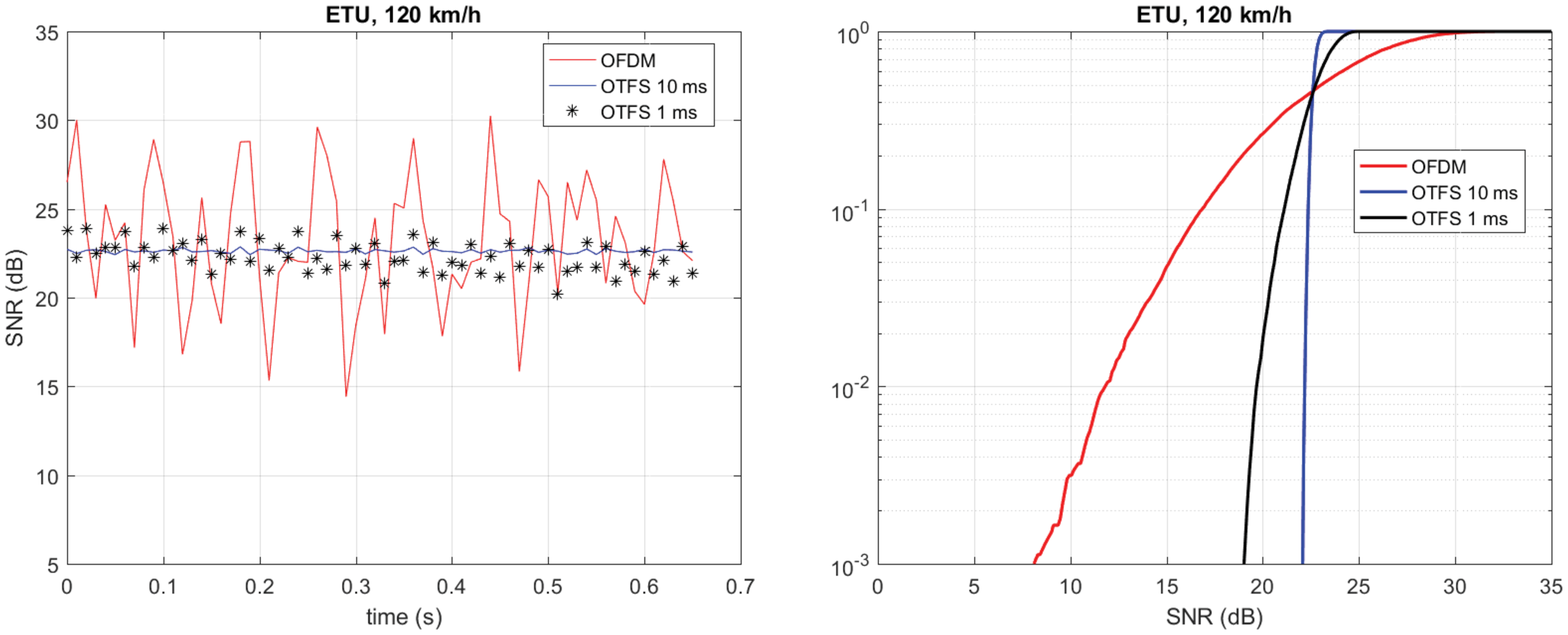}
  \caption{Evolution of receive SNR over time for (i) OFDM system, (ii) OTFS with 1ms window, and (iii) OTFS with 10 ms window. Channel is ETU, velocity of the UE is 120 km/h. }
  \label{fig:SNR_distribution}
\end{figure}

\section{Conclusion}\label{sec:conclusion}


In this paper we developed OTFS, a novel two-dimensional modulation scheme for wireless communications with significant advantages in performance over existing modulation schemes such as TDMA and OFDM. OTFS operates in the delay-Doppler coordinate system and we showed that with this modulation scheme coupled with equalization, all modulated symbols experience the same channel gain by extracting the full channel diversity. As a result of its operating principle, OTFS has the following important advantages:
\begin{itemize}
\item No need for channel adaptation, since OTFS provides a stable data rate. This is especially important in systems with high mobility, where feedback of CSI to the transmitter becomes impossible, or afflicted with large overhead. 
\item Better packet error rates (for the same SNR) or reduced SNR requirements (for the same PER) in the presence of high mobility (V2V, high-speed rail), or high phase noise (mm-wave systems). 
\item Improved PAPR, in particular for short packet transmission.
\item Improved MIMO capacity when using finite-complexity receivers. 
\end{itemize}

As a new modulation format, there are many aspects of OTFS that merit closer investigation, to optimize performance, reduce complexity, and enhance coexistence with existing systems. These aspects will be investigated in future publications. 


\appendix

Our goal is prove the formula \eqref{xkl} of Theorem \ref{Main_Thm}. Recall that $\hat{Y}[n,m]=Y[n,m]+V[n,m]$ where $Y[n,m]$ is the noise free term and $V[n,m]$ is additive noise. Substituting $Y[n,m]$ in the demodulation equation \eqref{eq:xYb} and using the time-frequency channel equation \eqref{eq:Y_H_X} and the modulation formula \eqref{eq:Xnm}, we can write: 
\begin{equation}
  \begin{split}
    &y[k,l] = \frac{1}{NM}\sum_{k^\prime=0}^{M-1} \sum_{l^\prime=0}^{N-1} x[k^\prime,l^\prime] \iint h_{\rm c}(\tau,\nu) e^{-j2\pi \nu \tau} \\
    &\times  \sum_{n,m=-\infty}^\infty W[n,m] e^{-j2\pi \left ( nT \left( \frac{l-l^\prime}{NT}-\nu \right ) - m \Delta f \left( \frac{k-k^\prime}{M\Delta f}-\tau \right) \right )}\mathrm{d}\tau \mathrm{d}\nu.
  \end{split}
\end{equation}
Since the factor in brackets is the discrete symplectic Fourier transform of $W[n,m]$ we have:
\begin{equation}
  \begin{split}
    &y[k,l] =\frac{1}{NM}\sum_{k^\prime=0}^{M-1} \sum_{l^\prime=0}^{N-1} x[k^\prime,l^\prime] \iint h_{\rm c}(\tau,\nu) e^{-j2\pi \nu \tau} \\
    &\times w\left( \frac{k-k^\prime}{M\Delta f}-\tau, \frac{l-l^\prime}{NT}-\nu \right) \mathrm{d}\tau \mathrm{d}\nu.
  \end{split}
\end{equation}
Further recognizing the double integral as a convolution of the channel impulse response (multiplied by an exponential) with the transformed window we obtain:
\begin{equation}
  y[k,l] =\frac{1}{NM}\sum_{k^\prime=0}^{M-1} \sum_{l^\prime=0}^{N-1} x[k^\prime,l^\prime] h_w \left( \frac{k-k^\prime}{M\Delta f}-\tau, \frac{l-l^\prime}{NT}-\nu \right),
\end{equation}
which is the desired result.

\bibliographystyle{IEEEtran}
\bibliography{IEEEabrv,otfslit}

\end{document}